\RequirePackage[2020-02-02]{latexrelease}
\documentclass[onecolumn,amsmath,amssymb,10pt,aps]{revtex4}
  
\usepackage[utf8]{inputenc}
\usepackage[T1]{fontenc}

\usepackage{amsmath}
\usepackage{amsfonts}
\usepackage{graphicx}
\usepackage{yfonts}
\usepackage{color}
\usepackage[normalem]{ulem}
\usepackage{amsthm}
\usepackage{bm}
\usepackage{bbm}
\usepackage{mathtools}
\usepackage{array}
\usepackage{placeins}
\usepackage{enumitem}
\usepackage{tikz}

\newcommand{\der}{\,\mathrm{d}}

\def\<{\langle}
\def\>{\rangle}

\newcommand{\Tr}{\mathrm{Tr}}
\def\oper{{\mathchoice{\rm 1\mskip-4mu l}{\rm 1\mskip-4mu l}
{\rm 1\mskip-4.5mu l}{\rm 1\mskip-5mu l}}}
\DeclareMathAlphabet\mathbfcal{OMS}{cmsy}{b}{n}

\mathchardef\mhyphen="2D 

\newtheorem{Theorem}{Theorem}

\newtheorem{Remark}{Remark}
\newtheorem{Proposition}{Proposition}
\newtheorem{Example}{Example}

\begin{document}

\title{Improving classical capacity of qubit dynamical maps\\through stationary state manipulation}
	
\author{Katarzyna Siudzi\'{n}ska}
\affiliation{Institute of Physics, Faculty of Physics, Astronomy and Informatics \\  Nicolaus Copernicus University in Toru\'{n}, ul. Grudzi\k{a}dzka 5, 87--100 Toru\'{n}, Poland}

\begin{abstract}
We analyze the evolution of Holevo and entanglement-assisted classical capacities for a special class of phase-covariant channels. In particular, we show that these capacities can be improved by changing the stationary state of the channel, which is closely related to its non-unitality degree. The more non-unital the channel, the greater its capacity. The channel parameters are engineered through mixtures on the level of dynamical maps, time-local generators, and memory kernels, for which we propose construction methods. For highly non-unital maps, we achieve a temporary increase in the classical capacity that exceeds the entanglement-assisted classical capacity of the unital map. This shows that non-unitality can become a better quantum resource for information transition purposes than quantum entanglement.
\end{abstract}

\flushbottom

\maketitle

\thispagestyle{empty}

\section{Introduction}

Establishing viable methods to reliably transmit information through quantum maps is an important task in quantum information processing and communication. However, the influence of external noise, which is often unavoidable, can be detrimental to many quantum information processing tasks. Therefore, it is necessary to protect information passed through a noisy channel. Approaches to error mitigation in quantum systems, also known as passive error correction, include decoherence-free subspaces (DFSs) \cite{Whaley}, noiseless subsystems \cite{Knill}, and dynamical decoupling \cite{Viola}.

Besides error correction, there are other ways to approach the problem of detrimental environmental noise effects on quantum systems. Another method is to accept the existence of noises and instead try to benefit from them. Verstraete et al. \cite{Verstraete} showed that local dissipation can be used as a resource for quantum computation and state engineering. For quantum information processing purposes, Marshall et al. \cite{Marshall} developed dissipation-assisted modular networks to control the loss of coherence and concurrence, whereas Gillard et al. \cite{Gillard} manipulated quantum thermal noise to enhance qubit state detection. Finally, memory effects associated with environmental noise were used to suppress error rates of quantum channels \cite{zanardi17,fidelity,Engineering_capacity}. An improved performance measured by the channel fidelity, output purity, and the ability to preserve quantum entanglement was attained through adding noises to the Markovian evolution on the level of memory kernel master equations \cite{Shabani,semi-2}. Recently, it has also been shown how to engineer non-unitality in order to enhance communication properties of qubit channels \cite{engineering_GAD}.

The maximal rate of information that can be reliably transmitted through a quantum channel is determined by the channel capacity. While classical channels have a unique definition of capacity, quantum channels can transmit information in a number of ways. A direct generalization of the Shannon capacity for classical channels is the classical capacity \cite{KingRemarks,Gyongyosi}, which measures the maximum rate of classical information sent with separable inputs. For any quantum channel $\Lambda$, it is defined via the asymptotic formula
\begin{equation}
C(\Lambda)=\lim_{n\to\infty}\frac 1n \chi(\Lambda^{\otimes n}),
\end{equation}
where the entropic expression
\begin{equation}
\chi(\Lambda)=\max_{\{p_k,\rho_k\}}\left[S\left(\sum_kp_k\Lambda[\rho_k]\right)
-\sum_kp_kS(\Lambda[\rho_k])\right]
\end{equation}
with the von Neumann entropy $S(\rho)=-\Tr(\rho\log_2\rho)$ is maximized over ensembles of separable states $\rho_k$ with probabilities $p_k$. In the above formula, $\chi(\Lambda)$ describes the capacity under a single use of the channel, also known as the Holevo capacity \cite{Holevo,sw}. Whenever the exact formula for the classical capacity is not known, $\chi(\Lambda)$ provides a lower bound, $\chi(\Lambda)\leq C(\Lambda)$. An upper bound can be found by calculating the entanglement-assisted classical capacity \cite{Bennett}
\begin{equation}
C_E(\Lambda)=\max_\rho I(\rho,\Lambda),
\end{equation}
where the sender and receiver share an unlimited amount of entanglement.
The maximalization is done over the mutual information
\begin{equation}
I(\rho,\Lambda)=S(\rho)+S(\Lambda[\rho])-S(\Lambda^c[\rho])
\end{equation}
with $\Lambda^c$ denoting the complementary channel \cite{ShorDevetak,Holevo_complement}. Other important channel capacities include the private classical capacity \cite{Devetak2} and quantum capacity \cite{Lloyd,Shor,Devetak2} -- for more information, see e.g. review works \cite{Gyongyosi,Smith}. Despite extensive research, not many computable expressions are known, even for qubit channels.

In this paper, we show that it is possible to increase the classical capacity of quantum channels via stationary state manipulation. The stationary states of generalized amplitude damping channels can be engineered by considering classical mixtures of legitimate memory kernels, for which we develop construction methods. Interestingly, it turns out that identical results are achievable if one instead mixes the corresponding time-local generators or dynamical maps. By increasing the non-unitality degree of dynamical maps, we prove that one obtains an absolute increase in both the Holevo and entanglement-assisted capacities. For non-unital channels, the classical capacity temporary exceeds even the entanglement-assisted classical capacity of unital channels. This proves that non-unitality is an important quantum resource for information transition purposes, as it can help to generate higher capacities than quantum entanglement.

\section{Phase-covariant channels}

Consider a class of non-unital qubit maps $\Lambda$ characterized by their covariance property
\begin{equation}\label{cov_def}
\Lambda\big[U(\phi)X U^\dagger(\phi)\big] = U(\phi)\Lambda[X]U^\dagger(\phi)
\end{equation}
that holds for all operators $X \in \mathcal{B}(\mathcal{H})$, $\mathcal{H}\simeq\mathbb{C}^2$, and real parameters $\phi$.
Such maps are covariant with respect to phase rotations represented by a unitary transformation
\begin{equation}
\label{eq:conds}
U(\phi)=\exp(-i\sigma_3\phi),\qquad\phi\in\mathbb{R},\qquad\sigma_3=\begin{pmatrix}
1 & 0\\0 & -1\end{pmatrix},
\end{equation}
which warrants the name {\it phase-covariant maps}. In addition, if $\Lambda$ are completely positive and trace-preserving, they are referred to as a {\it phase-covariant channels}. They are used to describe quantum evolution that combines pure dephasing with energy emission and absorption \cite{phase-cov-PRL,phase-cov}. Examples include thermalization and dephasing processes beyond the Markovian approximation \cite{PC1} and a weakly-coupled spin-boson model under the secular approximation \cite{PC3}.

The most general form of the phase-covariant channel reads, up to the unitary rotation $\rho\mapsto\exp(-i\sigma_3\theta)\rho\exp(i\sigma_3\theta)$
\cite{phase-cov,phase-cov-PRL},
\begin{equation}
\Lambda[\rho]=\frac 12 \left[(\mathbb{I}+\lambda_{\ast}\sigma_3)\Tr\rho
+\lambda_1\sigma_1\Tr(\rho\sigma_1)+\lambda_1\sigma_2\Tr(\rho\sigma_2)
+\lambda_3\sigma_3\Tr(\rho\sigma_3)\right],
\end{equation}
where $\sigma_\alpha$ denote the Pauli matrices. The real parameters $\lambda_1$, $\lambda_3$, $\lambda_{\ast}$ appear in the eigenvalue equations
\begin{equation}\label{rhoast}
\Lambda[\sigma_1]=\lambda_1\sigma_1,\qquad
\Lambda[\sigma_2]=\lambda_1\sigma_2,\qquad
\Lambda[\sigma_3]=\lambda_3\sigma_3,\qquad
\Lambda[\rho_\ast]=\rho_\ast=\frac{1}{2}\left[\mathbb{I}+\frac{\lambda_{\ast}}{1-\lambda_3}\sigma_3\right].
\end{equation}
The density operator $\rho_\ast$ is the stationary state of the channel, which means that it is preserved during the evolution. Note that for $\lambda_\ast=0$, one recovers the maximally mixed state $\rho_\ast=\mathbb{I}/2$, and so the associated $\Lambda$ is unital. Therefore, $\lambda_{\ast}$ controls its non-unitality property. Finally, $\lambda_1$, $\lambda_3$, $\lambda_\ast$ must satisfy the complete positivity conditions \cite{phase-cov}
\begin{equation}
|\lambda_\ast|+|\lambda_3|\leq 1,\qquad 4\lambda_1^2+\lambda_\ast^2\leq (1+\lambda_3)^2.
\end{equation}

To describe the time-evolution of open quantum systems, it is necessary to introduce a one-parameter family $\Lambda(t)$ of quantum channels that satisfy the initial condition $\Lambda(0)=\oper$. Now, $\Lambda(t)$ is called a {\it dynamical map}. The phase-covariant dynamical maps are solutions of the master equations $\dot{\Lambda}(t)=\mathcal{L}(t)\Lambda(t)$ with the time-local generators of the form
\begin{equation}
\mathcal{L}(t)=\gamma_+(t)\mathcal{L}_++\gamma_-(t)\mathcal{L}_-
+\gamma_3(t)\mathcal{L}_3,
\end{equation}
where
\begin{equation}\label{Lpm3}
\mathcal{L}_\pm[X]=\sigma_\pm X\sigma_\mp -\frac 12 \{\sigma_\mp\sigma_\pm,X\},
\qquad \mathcal{L}_3[X]=\frac 14(\sigma_3X\sigma_3-X).
\end{equation}
Note that $\mathcal{L}(t)$ includes, as special cases, the generators of amplitude damping ($\gamma_+(t)=\gamma_3(t)=0$), inverse amplitude damping ($\gamma_-(t)=\gamma_3(t)=0$), generalized amplitude damping ($\gamma_3(t)=0$), and pure dephasing ($\gamma_\pm(t)=0$).
We find the correspondence between the channel parameters and the decoherence rates $\gamma_\pm(t)$, $\gamma_3(t)$ \cite{phase-cov};
\begin{equation}
\lambda_1(t)=\exp\left\{-\frac 12 \Big[\Gamma_+(t)+\Gamma_-(t)+\Gamma_3(t)\Big]\right\},\qquad
\lambda_3(t)=\exp\Big[-\Gamma_+(t)-\Gamma_-(t)\Big],
\end{equation}
\begin{equation}
\lambda_\ast(t)=\exp\Big[-\Gamma_+(t)-\Gamma_-(t)\Big]\int_0^t
\Big[\gamma_+(\tau)-\gamma_-(\tau)\Big]\exp\Big[\Gamma_+(\tau)+\Gamma_-(\tau)\Big]
\der\tau,
\end{equation}
with $\Gamma_\mu(t)=\int_0^t\gamma_\mu(\tau)\der\tau$. For constant in time positive rates, $\Lambda(t)$ is the Markovian semigroup.

\section{Memory kernel master equations}

In an alternative approach, quantum evolution is provided by the Nakajima-Zwanzig equation \cite{Nakajima,Zwanzig}
\begin{equation}
\dot{\Lambda}(t)=\int_0^tK(t-\tau)\Lambda(\tau)\der\tau,
\end{equation}
where $K(t)$ is the memory kernel responsible for encoding memory effects. For the phase-covariant channels, a natural choice of the kernel is
\begin{equation}\label{kernel}
K(t)=k_+(t)\mathcal{L}_++k_-(t)\mathcal{L}_-+k_3(t)\mathcal{L}_3,
\end{equation}
where $\mathcal{L}_\pm$ and $\mathcal{L}_3$ are defined as in eq. (\ref{Lpm3}). Observe that $K(t)$ acts on the Pauli matrices as in
\begin{equation}
K(t)[\sigma_1]=\kappa_1(t)\sigma_1,\qquad K(t)[\sigma_2]=\kappa_1(t)\sigma_2,
\qquad K(t)[\sigma_3]=\kappa_3(t)\sigma_3,\qquad K(t)[\mathbb{I}]=\kappa_\ast(t)\sigma_3.
\end{equation}
The functions $\kappa_1(t)$, $\kappa_3(t)$, $\kappa_\ast(t)$ are related to the kernel parameters in the following way,
\begin{equation}\label{kappak}
\kappa_1(t)=-\frac 12 [k_+(t)+k_-(t)+k_3(t)],\qquad
\kappa_3(t)=-[k_+(t)+k_-(t)],\qquad
\kappa_\ast(t)=k_+(t)-k_-(t).
\end{equation}
Observe that $\kappa_\ast(t)$ depends on the difference of $k_\pm(t)$ rather than the sum. Therefore, it is possible to have $\kappa_\ast(t)=0$ (a unital channel) and $\kappa_\alpha(t)\neq 0$ by manipulating only a single kernel parameter. The inverse relation reads
\begin{equation}
k_\pm(t)=\frac{-\kappa_3(t)\pm\kappa_\ast(t)}{2},\qquad
k_3(t)=\kappa_3(t)-2\kappa_1(t).
\end{equation}
Now, instead of solving the dynamical equation for operators, one can equivalently solve a system of equations for functions,
\begin{equation}
\left\{
\begin{aligned}
\dot{\lambda}_1(t)&=\int_0^t\kappa_1(t-\tau)\lambda_1(\tau)\der\tau,\\
\dot{\lambda}_3(t)&\int_0^t\kappa_3(t-\tau)\lambda_3(\tau)\der\tau,\\
\dot{\lambda}_\ast(t)&=\int_0^t\kappa_3(t-\tau)\lambda_\ast(\tau)\der\tau+\int_0^t\kappa_\ast(\tau)\der\tau.
\end{aligned}
\right.
\end{equation}
Note that the integro-differential equation for $\lambda_\ast(t)$ is inhomogeneous. Moreover, it is $\kappa_\ast(t)$ that determines this inhomogenity, and the homogeneous part instead depends on $\kappa_3(t)$. In the Laplace transform domain, the solutions to the above system of dynamical equations are given by
\begin{equation}\label{LT}
\widetilde{\lambda}_1(s)=\frac{1}{s-\widetilde{\kappa}_1(s)},\qquad
\widetilde{\lambda}_3(s)=\frac{1}{s-\widetilde{\kappa}_3(s)},\qquad
\widetilde{\lambda}_\ast(s)=\frac{\widetilde{\kappa}_\ast(s)}{s}\frac{1}{s-\widetilde{\kappa}_3(s)}.
\end{equation}
From the properties of Laplace transforms, it follows that $\lambda_\ast(t)$ is related to $\lambda_3(t)$ via a convolution,
\begin{equation}
\lambda_\ast(t)=(K_\ast\ast\lambda_3)(t)=\int_0^tK_\ast(t-\tau)\lambda_3(\tau)\der\tau,
\qquad K_\ast(t)=\int_0^t\kappa_\ast(\tau)\der\tau.
\end{equation}
This formula is homogeneous, however the rght hand-side depends on $\lambda_3(t)$ instead of $\lambda_\ast(t)$.

Following the methods developed for the Pauli and generalized Pauli channels \cite{chlopaki,memory}, we parameterize the channel eigenvalues via
\begin{equation}
\lambda_j(t)=1-\int_0^t\ell_j(\tau)\der\tau,\qquad j=1,3.
\end{equation}
Additionally, we introduce the function $\ell_\ast(t)$ such that $\lambda_\ast(t)=-\int_0^t\ell_\ast(\tau)\der\tau$.

\begin{Theorem}\label{TH1}
The memory kernel $K(t)$ defined in eq. (\ref{kernel}) produces a legitimate phase-covariant dynamical map $\Lambda(t)$ if
\begin{equation}
\left\{
\begin{aligned}
\widetilde{\kappa}_1(s)&=-\frac{s\widetilde{\ell}_1(s)}{1-\widetilde{\ell}_1(s)},\\
\widetilde{\kappa}_3(s)&=-\frac{s\widetilde{\ell}_3(s)}{1-\widetilde{\ell}_3(s)},\\
\widetilde{\kappa}_\ast(s)&=-\frac{s\widetilde{\ell}_\ast(s)}{1-\widetilde{\ell}_3(s)},
\end{aligned}
\right.
\end{equation}
where
\begin{align}
\left|\int_0^t\ell_\ast(\tau)\der\tau\right|&\leq\int_0^t\ell_3(\tau)\der\tau,\label{C1}\\
\frac 12 \left(\int_0^t\ell_3(\tau)\der\tau
+\left|\int_0^t\ell_\ast(\tau)\der\tau\right|\right)
&\leq\int_0^t\ell_1(\tau)\der\tau\leq
2-\frac 12 \left(\int_0^t\ell_3(\tau)\der\tau
+\left|\int_0^t\ell_\ast(\tau)\der\tau\right|\right).\label{C2}
\end{align}
\end{Theorem}

\begin{proof}
The proof follows directly from the parameterization of $\lambda_j(t)$ and $\lambda_\ast(t)$, as well as the sufficient (but not necessary) linear conditions for the complete positivity of $\Lambda(t)$ \cite{phase-cov},
\begin{equation}
\lambda_3+|\lambda_\ast|\leq 1,\qquad 1-2|\lambda_1|+\lambda_3-|\lambda_\ast|\geq 0.
\end{equation}
\end{proof}

\begin{Example}
To illustrate our results, let us consider exponentially decaying functions $\ell_\mu(t)=\eta e^{-\xi_\mu t}$, $\mu\in\{1,3,\ast\}$, where $\xi_\ast\geq\xi_3\geq\xi_1\geq\eta\geq 0$.
It is straightforward to check that these functions satisfy the conditions from Theorem \ref{TH1}. Indeed, eq. (\ref{C1}) holds due to
\begin{equation}
h(\xi,t)=\frac{1}{\xi}\left(1-e^{-\xi t}\right)
\end{equation}
monotonically decreasing with the increase of $\xi$ for all $t\geq 0$ \cite{memory}. Using the same argumentation, one also proves eq. (\ref{C2}). The associated memory kernel $K(t)$ is defined via
\begin{align}
\kappa_1(t)&=\eta\left[-\delta(t)+(\xi_1-\eta)e^{-(\xi_1-\eta)t}\right],\\
\kappa_3(t)&=\eta\left[-\delta(t)+(\xi_3-\eta)e^{-(\xi_3-\eta)t}\right],\\
\kappa_\ast(t)&=\eta\left\{-\delta(t)+\frac{1}{\eta+\xi_\ast-\xi_3}
\left[\eta(\xi_3-\eta)e^{-(\xi_3-\eta)t}+\xi_\ast(\xi_\ast-\xi_3)e^{-\xi_\ast t}\right]\right\},
\end{align}
produces the phase-covariant dynamical map with
\begin{align}
\lambda_1(t)&=1-\frac{\eta}{\xi_1}\left(1-e^{-\xi_1t}\right),\\
\lambda_3(t)&=1-\frac{\eta}{\xi_3}\left(1-e^{-\xi_3t}\right),\\
\lambda_\ast(t)&=-\frac{\eta}{\xi_\ast}\left(1-e^{-\xi_\ast t}\right).
\end{align}
\end{Example}

A simplified case follows from a single-function parameterization
\begin{equation}
\ell_j(t)=1-\frac{1}{a_j}\int_0^t\ell(\tau)\der\tau,\qquad
\ell_\ast(t)=\pm\frac{1}{a_\ast}\int_0^t\ell(\tau)\der\tau.
\end{equation}

\begin{Proposition}
Take positive numbers $a_1,a_3,a_\ast$ and a function $\ell(t)$ such that $\int_0^t\ell(\tau)\der\tau\geq 0$. If
\begin{equation}
a_3\leq a_\ast,\qquad a_1\leq\frac{2a_3a_\ast}{a_3+a_\ast},\qquad
\int_0^t\ell(\tau)\der\tau\leq4\left(\frac{2}{a_1}+\frac{1}{a_3}+\frac{1}{a_\ast}\right)^{-1},
\end{equation}
then the corresponding memory kernel with
\begin{equation}
\left\{
\begin{aligned}
\widetilde{\kappa}_1(s)&=-\frac{s\widetilde{\ell}(s)}{a_1-\widetilde{\ell}(s)},\\
\widetilde{\kappa}_3(s)&=-\frac{s\widetilde{\ell}(s)}{a_3-\widetilde{\ell}(s)},\\
\widetilde{\kappa}_\ast(s)&=\pm\frac{sa_3\widetilde{\ell}(s)}{a_3a_\ast-a_\ast\widetilde{\ell}(s)}
\end{aligned}
\right.
\end{equation}
defines a legitimate phase-covariant qubit dynamical map characterized by
\begin{equation}
\lambda_j(t)=1-\frac{1}{a_j}\int_0^t\ell(\tau)\der\tau,\quad j=1,3,\qquad
\lambda_\ast(t)=\pm\int_0^t\ell_\ast(\tau)\der\tau.
\end{equation}
\end{Proposition}

\section{Improving classical capacities of dynamical maps}

In this section, we analyze how the choice of a stationary state of a quantum dynamical map, which is strictly related to its non-unitality degree, influences its classical capacities. First, let us justify the choice of a family of quantum channels. To exclude the influence of other effects, assume that the phase-covariant qubit channel $\Lambda$ has fixed eigenvalues and only the parameter $\lambda_\ast$ can be varied. According to ref. \cite{engineering_GAD}, there is a simple connection between $\lambda_\ast$ and the non-unitality measure
\begin{equation}
\mathrm{NU}(\Lambda)=\frac{|\lambda_\ast|}{1-|\lambda_3|},
\end{equation}
which determines the degree of non-unitality for phase-covariant $\Lambda$. It is easy to check that $\mathrm{NU}(\Lambda)=0$ and $\mathrm{NU}(\Lambda)=1$ correspond to unital and maximally non-unital maps, respectively.
Recall that $\lambda_\ast$ also appears in the definition of the stationary state $\rho_\ast$ in eq. (\ref{rhoast}), and hence the non-unitality degree $\mathrm{NU}(\Lambda)$ changes through the manipulations on $\rho_\ast$. Actually, unital channels are associated with maximally mixed stationary states, whereas maximally non-unital channels correspond to pure stationary states.

As a case study, consider the generalized amplitude damping channel, for which the Holevo and entanglement-assisted classical capacities are known. Let us introduce the following parametrization,
\begin{equation}\label{GADC}
\lambda_1=\lambda,\qquad \lambda_3=\lambda^2,\qquad \lambda_\ast=p(1-\lambda^2),
\end{equation}
with $-1\leq p\leq 1$, so that $\mathrm{NU}(\Lambda)=|p|$. Now, both the non-unitality degree and the stationary state $\rho_\ast={\rm{diag}}(1+p,1-p)/2$  are easily controlled via the parameter $p$. Notably, one recovers the amplitude damping channel for $p=-1$, inverse amplitude damping for $p=1$, and the unital channel for $p=0$.

Any generalized amplitude damping channel characterized by eq. (\ref{GADC}) can be constructed through a classical mixture of two quantum channels. Recall that if two dynamical maps are physically admissible, then so is their convex combination. Now, denote the unital (Pauli) channel by $\Lambda_{\rm U}$ and the maximally non-unital channel corresponding to $p=\pm 1$ by $\Lambda_{\rm NU}^{\pm}$. Assume that $\Lambda_{\rm U}$ and $\Lambda_{\rm NU}$ have all common eigenvalues and three common eigenvectors, as in
\begin{equation}
\Lambda_{\rm U}[\sigma_k]=\lambda_k\sigma_k,\qquad 
\Lambda_{\rm NU}^{\pm}[\sigma_k]=\lambda_k\sigma_k,\qquad k=1,2,3\qquad (\lambda_2\equiv\lambda_1).
\end{equation}
The final eigenvectors follow from
\begin{equation}
\Lambda_{\rm U}[\mathbb{I}]=\mathbb{I},\qquad
\Lambda_{\rm NU}^+[|0\>\<0|]=|0\>\<0|,\qquad
\Lambda_{\rm NU}^-[|1\>\<1|]=|1\>\<1|.
\end{equation}
Now, the quantum channels resulting from the mixtures
\begin{equation}\label{mixtures}
\Lambda^{\pm}=(1-|p|)\Lambda_{\rm U}+|p|\Lambda_{\rm NU}^{\pm}
\end{equation}
cover the entire class of phase-covariant channels defined in eq. (\ref{GADC}). Experimentally, they can be realized using the Mach-Zender interferometer, where $\Lambda_{\rm U}$ and $\Lambda_{\rm NU}^{\pm}$ operate on their own paths, and $p$ is characterized by the properties of a beam splitter \cite{Uriri,Siltanen}.

In the following subsections, we analyze capacities of dynamical maps of the form given in eq. (\ref{mixtures}). In other words, we consider
\begin{equation}\label{mixtures2}
\Lambda^{\pm}(t)=(1-|p|)\Lambda_{\rm U}(t)+|p|\Lambda_{\rm NU}^{\pm}(t)
\end{equation}
such that
\begin{equation}\label{eig}
\lambda_1(t)=\lambda(t),\qquad \lambda_3(t)=\lambda^2(t),\qquad \lambda_\ast(t)=p(1-\lambda^2(t)),
\end{equation}
and $\lambda(0)=1$. This way, the entire time-dependence is encoded in the function $\lambda(t)$, and the parameter $p$ remains constant in time.

The master equations for $\Lambda^{\pm}(t)$ display some interesting properties.

\begin{Proposition}
The mixtures of dynamical maps from eq. (\ref{mixtures2}) admit the master equations
\begin{equation}
\dot{\Lambda}^{\pm}(t)=\mathcal{L}^{\pm}(t)\Lambda^{\pm}(t)
\end{equation}
with time-local generators
\begin{equation}\label{mixL}
\mathcal{L}^{\pm}(t)=(1-|p|)\mathcal{L}_{\rm U}(t)+|p|\mathcal{L}_{\rm NU}^{\pm}(t)
\end{equation}
that are analogical mixtures of the generators $K_{\rm U}(t)$ and $K_{\rm NU}^{\pm}(t)$ for $\Lambda_{\rm U}(t)$ and $\Lambda_{\rm NU}^{\pm}(t)$, respectively.
\end{Proposition}

\begin{proof}
First, observe that $\Lambda_{\rm U}(t)$ and $\Lambda_{\rm NU}^{\pm}(t)$ are generated via
\begin{equation}
\mathcal{L}_{\rm U}(t)=-\frac{\dot{\lambda}(t)}{\lambda(t)}(\mathcal{L}_+
+\mathcal{L}_-),\qquad
\mathcal{L}_{\rm NU}^{\pm}(t)=-2\frac{\dot{\lambda}(t)}{\lambda(t)}\mathcal{L}_{\pm},
\end{equation}
respectively. These time-local generators can be mixed according to eq. (\ref{mixL}) into
\begin{equation}
\mathcal{L}(t)=-\frac{\dot{\lambda}(t)}{\lambda(t)}\left[(1+p)\mathcal{L}_++(1-p)\mathcal{L}_-\right]
\end{equation}
which indeed gives rise to the dynamical map $\Lambda(t)$ with eigenvalues as in eq. (\ref{eig}).
\end{proof}

In particular, note that a mixture of two Markovian semigroups leads back to a semigroup for any value of $p$.

Time-local generators obtainable through classical mixtures of physical generators generalize the notion of additive generators $\mathcal{L}(t)=\alpha\mathcal{L}^{(1)}(t)+\beta\mathcal{L}^{(2)}(t)$, which must be legitimate for any $\alpha,\beta\geq 0$ \cite{Kolodynski}. It has been shown that additive generators provide a physically valid evolution for either Markovian or commutative, semigroup-simulable $\mathcal{L}^{(k)}(t)$. However, such generators correspond to the dynamics derived microscopically from the system interacting with multiple environments if and
only if the cross-correlations between environments can be ignored, which is
the case for weak-coupling regime.

\begin{Proposition}
The mixtures of dynamical maps from eq. (\ref{mixtures2}) admit the memory kernel master equations
\begin{equation}
\dot{\Lambda}^{\pm}(t)=\int_0^tK^{\pm}(t-\tau)\Lambda^{\pm}(\tau)\der\tau,
\end{equation}
where the memory kernels
\begin{equation}\label{mixK}
K^{\pm}(t)=(1-|p|)K_{\rm U}(t)+|p|K_{\rm NU}^{\pm}(t)
\end{equation}
are analogical mixtures of the memory kernels $K_{\rm U}(t)$ and $K_{\rm NU}^{\pm}(t)$ that generate $\Lambda_{\rm U}(t)$ and $\Lambda_{\rm NU}^{\pm}(t)$, respectively.
\end{Proposition}

\begin{proof}
It is straightforward to check that $\Lambda_{\rm U}(t)$ and $\Lambda_{\rm NU}^{\pm}(t)$ arise from the master equations with the corresponding memory kernels (in the Laplace transform domain)
\begin{equation}
\widetilde{K}_{\rm U}(s)=-\frac{\widetilde{\kappa}_3(s)}{2}
(\mathcal{L}_++\mathcal{L}_-)+[\widetilde{\kappa}_3(s)-2\widetilde{\kappa}_1(s)]
\mathcal{L}_3,\qquad
\widetilde{K}_{\rm NU}^{\pm}(s)=-\widetilde{\kappa}_3(s)\mathcal{L}_{\pm}
+[\widetilde{\kappa}_3(s)-2\widetilde{\kappa}_1(s)]\mathcal{L}_3,
\end{equation}
where $\widetilde{\kappa}_3(s)=s-1/\widetilde{\lambda^2}(s)$ and $\widetilde{\kappa}_1(s)=s-1/\widetilde{\lambda}(s)$. Now, eq. (\ref{mixK}) corresponds to
\begin{equation}
\widetilde{K}(s)=-\frac{\widetilde{\kappa}_3(s)}{2}\left[(1+p)\mathcal{L}_++(1-p)\mathcal{L}_-\right]
+[\widetilde{\kappa}_3(s)-2\widetilde{\kappa}_1(s)]\mathcal{L}_3,
\end{equation}
which is exactly the memory kernel for the dynamical map $\Lambda(t)$ with eigenvalues given in eq. (\ref{eig}).
\end{proof}

Observe that the memory kernels possess an additional term proportional to $\mathcal{L}_3$, which is not present in the formulas for time-local generators. This term vanishes for Markovian semigroups. Incorporating non-local noises through memory kernel addition has been analyzed e.g. in refs. \cite{Marshall,fidelity,Engineering_capacity,engineering_GAD}. However, physical justification for such addition has not yet been studied on a microscopical level.

Based on eq. (\ref{mixtures2}), as well as Propositions 1 and 2, we make the following observation.

\begin{Remark}
The generalized amplitude damping dynamical maps defined in eq. (\ref{eig}) arise from classical mixtures taken on the level of dynamical maps, time-local generators, or memory kernels.
\end{Remark}

In the following subsections, we analyze the Holevo capacity, classical capacity, and entanglement-assisted classical capacity of generalized amplitude damping dynamical maps characterized by eq. (\ref{eig}). In particular, we focus on their time-evolution and behaviour under manipulations of the parameter $p$. We consider only mixtures with non-negative values of $p$, as both the Holevo and entanglement-assisted capacity are symmetric with respect to $p\mapsto-p$ \cite{GADC_Holevo,CE_GADC}.

\subsection{Holevo capacity}

The Holevo capacity is the maximal rate of classical information reliably transmittable under a single use of a quantum channel \cite{Holevo,sw}. For unital channels $\Lambda_{\rm U}$, a simple analytical formula is known \cite{KingQubit,Holevo_capacity},
\begin{equation}
\chi(\Lambda_{\rm U})=\frac{1+\lambda}{2}\log_2(1+\lambda)
+\frac{1-\lambda}{2}\log_2(1-\lambda).
\end{equation}
For the generalized amplitude damping channels, it has been shown that \cite{GADC_Holevo}
\begin{equation}
\chi(\Lambda)=\frac 12 [f(r)-f(q)].
\end{equation}
In our case, the function $f(x)$ and the parameter $r$ are defined by
\begin{equation}
f(x)=(1+x)\log_2(1+x)+(1-x)\log_2(1-x),\qquad
r=\sqrt{\lambda^2+q^2-\left(\frac{q-p}{\lambda}+p\lambda\right)^2},
\end{equation}
whereas the parameter $q$ is the solution of the implicit equation
\begin{equation}
f^\prime(r)(q-p)(1-\lambda^2)=-r\lambda^2f^\prime(q).
\end{equation}
The prime in $f^\prime(x)$ denotes the derivative over $x$.
For non-unital channels, only numerical solutions exist.

In Fig.1, we plot the evolution of Holevo capacity for $\lambda(t)=e^{-t}$ (left) and $\lambda(t)=\cos t$ (right) for various values of $p$. Due to the symmetries of the cosine function, the corresponding $\chi[\Lambda(t)]$ oscillates between $0$ and $1$ with the same periodicity. Observe that in both cases the Holevo capacity increases with $p$. In particular, because the Holevo capacity is equal to the classical capacity for unital maps $\Lambda_{\rm U}(t)$ (to be discussed in detail in Section 4.3), we can already see that $C[\Lambda(t)]$ increases with the degree of non-unitality at any $t>0$.

\FloatBarrier

\begin{figure}[htb!]
   \includegraphics[width=0.8\textwidth]{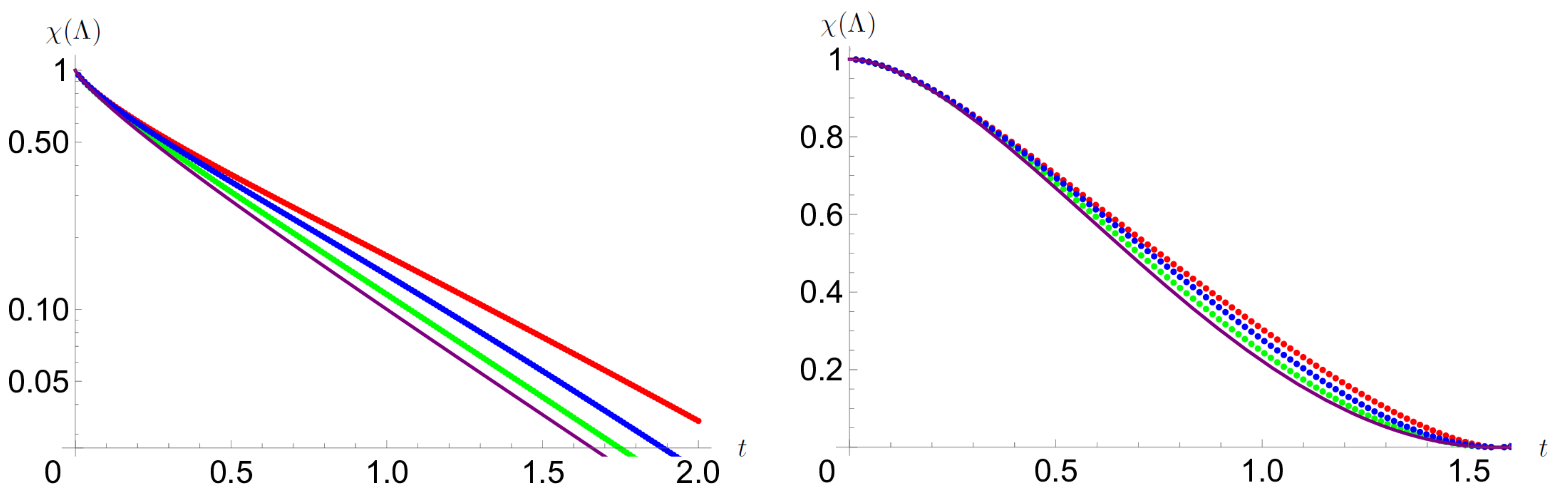}
\caption{Plots representing time-evolution of the Holevo capacity for exponentially decaying (left; logarithmic scale) and oscillating (right) function $\lambda(t)$. The color curves correspond to the parameter $p=0$ (purple), $p=2/3$ (green), $p=0.9$ (blue), and $p=1$ (red).}
\label{HC}
\end{figure}

\FloatBarrier

\subsection{Entanglement-assisted classical capacity}

If the classical information is sent under infinitely many uses of a quantum channel and the input states are allowed to be entangled, the maximal rate of reliably shared information is quantified by entanglement-assisted classical capacity \cite{Bennett}. An analytical formula has been recently found for the unital phase-covariant channels \cite{Unital_PCC},
\begin{equation}
C_E(\Lambda_{\rm U})=(1+\lambda)\log_2(1+\lambda)
+(1-\lambda)\log_2(1-\lambda)=2\chi(\Lambda_{\rm U}).
\end{equation}
The same relation was numerically proven for non-unital generalized amplitude damping channels \cite{CE_GADC}. Actually, for $\Lambda_{\rm U}$, the Holevo capacity is equal to its classical capacity $C(\Lambda_{\rm U})$ (to be discussed in detail in Section 4.3). Therefore, unital phase-covariant channels are another example for which $C_E(\Lambda)=2C(\Lambda)$. This exact relation was shown to hold for noiseless channels \cite{Wiesner} and quantum erasure channels \cite{Bennett2}. For generalized amplitude damping channels with $0<|p|<1$, the entanglement-assisted classical capacity is given by \cite{CE_GADC}
\begin{equation}
C_E(\Lambda)=\max_{-1\leq z\leq 1}F(\Lambda,z).
\end{equation}
The maximization is performed over
\begin{equation}
F(\Lambda,z):=H_2\left(\frac{1+z}{2}\right)+H_2\left(\frac{1+p+\lambda^2(z-p)}{2}\right)
+h_+\log_2h_++h_-\log_2h_-+\Delta_+\log_2\Delta_++\Delta_-\log_2\Delta_-,
\end{equation}
where $H_2(x):=-x\log_2 x-(1-x)\log_2(1-x)$ and
\begin{equation}
h_\pm=\frac{1\pm z}{4}(1-\lambda^2)(1\mp p),\qquad
\Delta_\pm=\frac 14 \left[1+\lambda^2+zp(1-\lambda^2)\pm\sqrt{4(\lambda^2+zp(1-\lambda^2))+(1-\lambda^2)^2
(p-z)^2}\right].
\end{equation}
If $p=\pm 1$, then the formula for the amplitude damping channel can be used \cite{ADC_quantum},
\begin{equation}
C_E(\Lambda_{\rm NU})=\max_{0\leq\pi\leq 1}
\left[H_2(\pi)+H_2(\pi\lambda^2)-H_2(\pi(1-\lambda^2))\right].
\end{equation}
Again, for non-unital channels, only numerical solutions have been found.

In Fig.2, we plot the evolution of entanglement-assisted classical capacity for $\lambda(t)=e^{-t}$ (left) and $\lambda(t)=\cos t$ (right) for various values of $p$. Again, the capacity for $\lambda(t)=\cos t$ oscillates between $0$ and $2$ with the same periodicity as the cosine function. Also, similarly to the Holevo capacity, $C_E[\Lambda(t)]$ increases along with the non-unitality degree $p$.

\FloatBarrier

\begin{figure}[htb!]
   \includegraphics[width=0.8\textwidth]{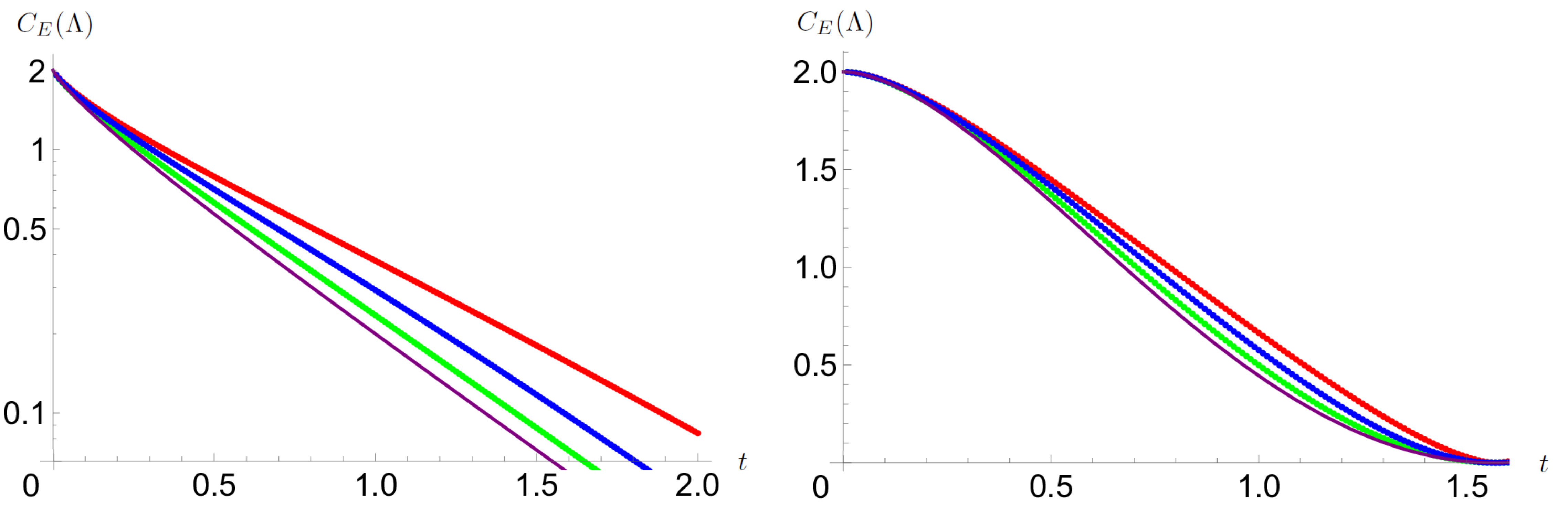}
\caption{Plots representing time-evolution of the entanglement-assisted classical capacity for exponentially decaying (left; logarithmic scale) and oscillating (right) function $\lambda(t)$. The color curves correspond to the parameter $p=0$ (purple), $p=2/3$ (green), $p=0.9$ (blue), and $p=1$ (red).}
\label{EA}
\end{figure}

\FloatBarrier

\subsection{Classical capacity}

The classical capacity measures the maximal rate of classical information reliably transmitted under infinitely many uses of a quantum channel with separable inputs only. It is related to the Holevo capacity via the asymptotic expression \cite{Holevo_CC,sw,Holevo}
\begin{equation}
C(\Lambda)=\lim_{n\to\infty}\frac 1n \chi(\Lambda^{\otimes n}).
\end{equation}
For the Pauli channels, the Holevo capacity is weakly additive, so that $\chi(\Lambda_U\otimes\Lambda_U)=2\chi(\Lambda_U)$. Therefore, unital phase-covariant channels satisfy \cite{Holevo_capacity}
\begin{equation}\label{ineq1}
C(\Lambda_{\rm U})=\chi(\Lambda_{\rm U}).
\end{equation}
In general, however, $\chi(\Lambda\otimes\Lambda)\geq 2\chi(\Lambda)$, so it is usually very hard to find analytical formulas for the classical capacity. In such cases, one can resort to calculating its lower and upper bounds. Note that the following inequalities always hold,
\begin{equation}\label{ineq2}
\chi(\Lambda)\leq C(\Lambda)\leq C_E(\Lambda).
\end{equation}
In the previous subsections, we have already seen that the capacities for mixtures $\Lambda(t)$ of $\Lambda_{\rm U}(t)$ and $\Lambda_{\rm NU}(t)$ satisfy
$\chi[\Lambda_{\rm U}(t)]\leq\chi[\Lambda_{\rm NU}(t)]$ as well as $C_E[\Lambda_{\rm U}(t)]\leq C_E[\Lambda_{\rm NU}(t)]$. This, together with eqs. (\ref{ineq1}) and (\ref{ineq2}), introduces the following partial orders,
\begin{align}
&\chi[\Lambda_{\rm U}(t)]=C[\Lambda_{\rm U}(t)]\leq\chi[\Lambda(t)]
\leq C[\Lambda(t)]\leq C_E[\Lambda(t)],\\
&C_E[\Lambda_{\rm U}(t)]\leq C_E[\Lambda(t)],
\end{align}
at any given time. Therefore, it remains to check whether it is possible that $C[\Lambda(t)]\geq C_E[\Lambda_{\rm U}(t)]$.

In Fig.3, we compare the entanglement-assisted classical capacity for the unital dynamical map with the Holevo capacity $\chi[\Lambda(t)]$ for the non-unital maps with different degrees of non-unitality. Again, we choose exponentially decaying $\lambda(t)=e^{-t}$ (left) and oscillating $\lambda(t)=\cos t$ (right). 
Observe that there exist times $t$ at which $C_E[\Lambda_{\rm U}(t)]\leq \chi[\Lambda_{\rm NU}(t)]$. In other words, the lower bound for the classical capacity of a non-unital channel temporarily exceeds entanglement-assisted capacity for a unital channel. Hence, for classical information transition purposes, non-unitality turns out to be locally a better resource than entanglement.

\begin{figure}[htb!]
   \includegraphics[width=0.8\textwidth]{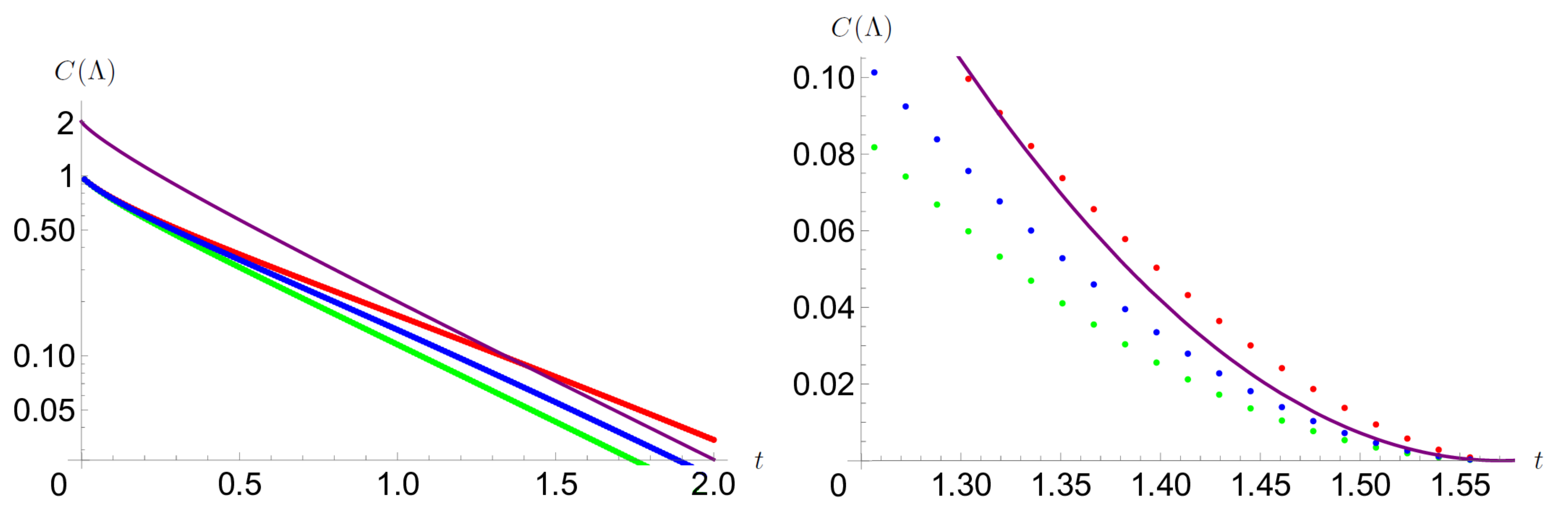}
\caption{Plots representing time-evolution of classical capacities for exponentially decaying (left) and oscillating (right) functions $\lambda(t)$. The color curves correspond to the entanglement-assisted capacity for $p=0$ (purple), as well as to the Holevo capacity for $p=2/3$ (green), $p=0.9$ (blue), and $p=1$ (red).}
\label{EA}
\end{figure}

\section{Conclusions}

We analyzed the evolution of phase-covariant dynamical maps under the master equations with time-local generators and non-local memory kernels, for which we provided methods of construction. As a case study, we considered classical mixtures of unital and maximally non-unital generalized amplitude damping channels. It turned out that such mixtures can also be obtained through convex combinations of the corresponding time-local generators and memory kernels. By changing the mixing parameter, it was possible to easily manipulate both the non-unitality degree and the stationary state of the dynamical map. We proved that the more non-unital the channel, the higher its Holevo and entanglement-assisted capacities. The increase in Holevo capacity was so significant that the classical capacity of highly non-unital channels temporary exceeded the entanglement-assisted classical capacity of unital channels. Therefore, non-unitality was more successful in enhancing the classical channel capacity than quantum entanglement.

Further analysis for the full family of phase-covariant dynamical maps is required to better understand this phenomenon. However, this first requires solving the open questions of finding analytical solutions for their capacities.

\section{Acknowledgements}

This research was funded in whole or in part by the National Science Centre, Poland, Grant number 2021/43/D/ST2/00102. For the purpose of Open Access, the author has applied a CC-BY public copyright licence to any Author Accepted Manuscript (AAM) version arising from this submission.

\bibliography{C:/Users/cyndaquilka/OneDrive/Fizyka/bibliography}
\bibliographystyle{C:/Users/cyndaquilka/OneDrive/Fizyka/beztytulow2}

\end{document}